\documentclass[11pt,letterpaper]{article}

\usepackage{fullpage}
\usepackage{xcolor}
\usepackage[hidelinks,colorlinks = true,
           linkcolor = blue,
            urlcolor  = blue,
            citecolor = blue,
            anchorcolor = blue,
            bookmarks=false]{hyperref}
\usepackage{graphicx,subcaption}
\usepackage{enumitem}
\usepackage{algorithm}
\usepackage{algpseudocode}
\usepackage{amsmath}
\usepackage{amssymb}
\usepackage{amsthm}
\usepackage{tikz}
\usepackage[labelfont=bf]{caption}
\usepackage{multicol}
\usepackage{multirow}
\usepackage{array}
\usepackage{booktabs}
\usepackage{longtable}
\usepackage{tabularx}

\usepackage{amsmath,amsthm}
\newtheorem{theorem}{Theorem}
\newtheorem{corollary}{Corollary}[theorem] 
\newcolumntype{L}[1]{>{\raggedright\let\newline\\\arraybackslash\hspace{0pt}}m{#1}}
\newcolumntype{C}[1]{>{\centering\let\newline\\\arraybackslash\hspace{0pt}}m{#1}}
\newcolumntype{R}[1]{>{\raggedleft\let\newline\\\arraybackslash\hspace{0pt}}m{#1}}

\newtheorem{lemma}[theorem]{Lemma}

\newcommand{\tool}{\ensuremath{\mathsf{EntropyEstimation}}}
\newcommand{\prob}{\ensuremath{\mathsf{Pr}}}
\newcommand{\satisfying}[1]{\ensuremath{sol({#1})}} %

\newcommand{\entropy}[1]{\ensuremath{H({#1})}}

\newcommand{\supp}{\mathrm{supp}}

\newcommand{\expect}{\ensuremath{\mathsf{E}}}
\newcommand{\variance}{\ensuremath{\mathsf{variance}}}

\newcommand{\EntropyEstimation}{\ensuremath{\mathsf{EntropyEstimation}}}
\newcommand{\SampleEst}{\ensuremath{\mathsf{SampleEst}}}

\newcommand{\proc}{\ensuremath{\mathsf{PROC}}}

\usepackage[draft]{todonotes}   %
\algnotext{EndIf} 
\algnotext{EndFor}
\usepackage{makecell}
\begin{document}
\title{A Scalable Shannon Entropy Estimator\thanks{ {\EntropyEstimation} is available at \url{https://github.com/meelgroup/entropyestimation}. A preliminary version of this work appears at International Conference on Computer-Aided Verification, CAV, 2022.  The names of authors are sorted alphabetically and the order does not reflect contribution.}}
 \author{Priyanka Golia\\
 Indian Institute of Technology Kanpur\\
 National University of Singapore
 \and
 Brendan Juba\\
 Washington University in St. Louis\\
  \and 
 Kuldeep S. Meel\\
 National University of Singapore}
\date{}
\maketitle              %
\begin{abstract}
We revisit the well-studied problem of estimating the Shannon entropy of a 
probability distribution, 
now given access to a \emph{probability-revealing conditional 
sampling} oracle. In this model, the oracle takes as input the representation
of a set $S$, and returns a sample from the distribution obtained by 
conditioning on $S$, together with the probability of that sample in the
distribution.
Our work is motivated by applications of such algorithms in Quantitative 
Information Flow analysis (QIF) in programming-language-based security. 
Here, information-theoretic quantities capture the effort required on the part
of an adversary to obtain access to confidential information. These
applications demand accurate measurements when the entropy is small. Existing
algorithms that do not use conditional samples require a number of queries
that scale inversely with the entropy, which is unacceptable in this regime,
and indeed, a lower bound by Batu et al.\ (STOC 2002) established that no algorithm using 
only sampling and evaluation oracles can obtain acceptable performance. On the other hand, prior work in the conditional
sampling model by Chakraborty et al.\ (SICOMP 2016) only obtained a high-order polynomial query complexity,
$\mathcal{O}(\frac{m^7}{\epsilon^8}\log\frac{1}{\delta})$ queries, to obtain
additive $\epsilon$-approximations on a domain of size $\mathcal{O}(2^m)$;
note furthermore that additive approximations are also unacceptable for such
applications. No prior work could obtain polynomial-query multiplicative approximations to the entropy in the low-entropy regime.

We obtain multiplicative $(1+\epsilon)$-approximations using only
$\mathcal{O}(\frac{m}{\epsilon^2}\log\frac{1}{\delta})$ queries to the
probability-revealing conditional sampling oracle. Indeed, moreover, we obtain
small, explicit constants, and demonstrate that our algorithm obtains a
substantial improvement in practice over the previous state-of-the-art methods
used for entropy estimation in QIF.
\end{abstract}

\section{Introduction}
We consider the problem of estimating the entropy of a probability 
distribution $D$ over a discrete domain $\Omega$ of size $2^m$.  
Motivated by applications in \emph{quantitative information flow
analysis (QIF)}, a rigorous approach to quantitatively measure confidentiality
\cite{BKR09,S09,VEB+16}, we seek multiplicative estimates of the (Shannon) 
entropy in the low-entropy regime \cite{BKR09,CCH11,PMTP12,S09}. Indeed, in 
QIF, one would ideally like to certify that the information leakage is 
exponentially small; even a simple password checker that reports ``incorrect 
password'' leaks such an exponentially small amount of information about the 
password.

It is immediate that mere sample access to the distribution is inadequate for
any efficient algorithm to certify that the entropy is so small: distributions
with entropies that are vastly different in the multiplicative sense may 
nevertheless have negligible statistical distance and thus be 
indistinguishable (cf.~Batu et al.~\cite{BDKR05} and, in particular, Valiant 
and Valiant~\cite{VV11} for strong lower bounds). We thus consider a \emph
{conditional sampling} oracle, as introduced by Chakraborty et al.~\cite
{CFGM16} and independently by Canonne et al.~\cite{CRS15}, with \emph
{probability revealing samples}, as introduced by Onak and Sun~\cite{OS18}. A 
conditional sampling oracle, COND, for a distribution $D$ takes a representation of a
set $S$ and returns a sample drawn from $D$ conditioned on $S$. To extend the 
oracle to have probability-revealing samples means that in addition to the 
sample $x$, we obtain the probability of $x$ in $D$. Note that this
oracle, referred to as {\proc} henceforth, can be simulated by an evaluation oracle for $D$ together with a 
conditional sampling oracle for $D$.\footnote{%
An oracle that returns the \emph{conditional probability} of $x$ in $D$
conditioned on $S$ might seem at least as natural, but it is not clear whether
such an oracle can be simulated by the usual evaluation and conditional 
sampling oracles. In any case, our algorithm can be easily adapted to this 
alternative model.}

Using probability-revealing samples (described as a combined 
sampling-and-evaluation model), Guha et al.~\cite{GMGV09} obtained a
$\mathcal{O}(\frac{m}{\epsilon^2H}\log\frac{1}{\delta})$-query algorithm for
multiplicative $(1+\epsilon)$-approximations of the entropy $H$ for 
distributions on a $2^m$-element domain, with confidence $1-\delta$, which is
optimal in this model:~\cite{BDKR05} observe that this indeed scales badly for
exponentially small $H$. In the conditional sampling model, Chakraborty et 
al.\ obtained a $\mathcal{O}(\frac{1}{\epsilon^8}m^7\log\frac{1}
{\delta})$-query algorithm, for an \emph{additive} $\epsilon$-approximation of
the entropy. Note that when the entropy is so small, such additive 
approximations would require prohibitively small values of $\epsilon$ to 
provide useful estimates. In summary, all previous algorithms either
\begin{itemize}
\item used a superpolynomial number of queries in the bit-length of the elements $m$,
\item used a number of queries scaling with $1/H$, or
\item obtained additive estimates
\end{itemize} 
thus rendering them incapable of obtaining useful estimates in the low-entropy regime.

\subsection{Our contribution}

The primary contribution of our work is the first algorithm to obtain 
$(1+\epsilon)$-multiplicative estimates using a polynomial number of queries in the bit-length $m$ and approximation parameter $\epsilon$, with no dependence on $H$. Indeed, we use only $\mathcal{O}(\frac{m}{\epsilon^2}\log\frac{1}{\delta})$ samples given access  to {\proc}. 
Moreover, we obtain explicit 
constant factors that are sufficiently small that our algorithms are useful in
practice: we have experiments demonstrating that our algorithm can obtain 
estimates for benchmarks that are far beyond the reach of the existing tools 
for computing the entropy.

Our algorithm is a simple median-of-means estimator. To obtain multiplicative
estimates, we use second-moment methods, hence we need bounds on the ratio of 
the variance of the self-information to the square of the entropy. Indeed, 
Batu et al.~\cite{BDKR05} considered an approach that is similarly based on 
the second moments, which encounters two main issues: The first issue, as 
discussed above, is that if the entropy is small, this ratio may be very 
large. The second issue is that bounding the variance and the entropy 
separately is not sufficient to obtain the linear dependence on $m$: the 
variance of the self-information may be quadratic in $m$, and this is tight, 
as also shown by Batu et al. Our main technical contribution thus lies in how 
we bound this ratio: we use tight, explicit bounds in the high-entropy regime 
that obtain a linear dependence on $m$, together with a ``win-win'' strategy 
for using the conditional samples. Namely, we observe that when we condition 
on avoiding the high-probability element that dominates the distribution, 
either we obtain a conditional distribution with high entropy, in which case
we can use the aforementioned bounds, or else -- observing that the 
self-information w.r.t.\ the original distribution is quite large for all 
such elements -- we can obtain a bound on the variance directly that is 
similarly small, and in particular has the same linear dependence on $m$. We 
remark that while Guha et al.~\cite{GMGV09} obtained estimates with the same 
linear dependence on $m$ in the high-entropy case, they acheived this by 
dropping any samples that have high self-information and applying a Chernoff 
bound to the remaining, bounded samples. Since in the low-entropy case the 
samples generally have high self-information, it is unclear how one would 
extend their technique to handle the low-entropy case as we do here.

It remains an interesting open question whether or not our algorithm is 
optimal. Chakraborty et al.\ obtained a $\Omega(\sqrt{\log m})$ lower bound
for the conditional sampling model; we are not aware of any lower bounds for
the combined, conditional probability-revealing sampling model. Also, Acharya 
et al.~\cite{ACK18} obtained a $\tilde{\mathcal{O}}(\frac{\log m}
{\epsilon^3})$-query algorithm for the related problem of 
$(1+\epsilon)$-multiplicative support size estimation in the conditional 
sampling model. Support size estimation is generally easier than entropy 
estimation given access to an evaluation oracle -- indeed, additive 
$\epsilon\cdot 2^m$-approximations are possible with only $\mathcal{O}(\frac
{1}{\epsilon^2})$ queries -- but they suffer from similar issues with 
distributions with a light ``tail'' (cf.\ Goldreich~\cite{G19}). Conceivably,
conditional sampling might enable a similarly substantial reduction in the
query complexity of entropy estimation as well. 

\subsection{On the application to Quantitative Information Flow}\label{sec:qif-intro}
As mentioned at the outset, our work is motivated by the needs of quantitative
information flow (QIF) applications. It is therefore an important question
whether the {\proc} oracle model is realistic. To this end, we demonstrate that {\proc} can indeed be efficiently implemented using the available tools in automated reasoning, and our technique can be employed in such QIF analyses.

The standard recipe for using the QIF framework is to measure the information leakage from an underlying program $\Pi$ as follows. In a simplified model, a program $\Pi$ maps a set of controllable inputs ($C$) and secret inputs ($I$) to outputs ($O$) observable to an attacker. The attacker is interested in inferring $I$ based on the output $O$. 
A diverse array of approaches have been proposed to efficiently model $\Pi$, with techniques relying on a combination of symbolic analysis~\cite{PMTP12}, static analysis~\cite{CHM07}, automata-based techniques~\cite{ABB15,AEL+18,B19}, SMT-based techniques~\cite{PM14}, and the like. For each, the core underlying technical problem is to determine the leakage of information for a given observation. We often capture this leakage using entropy-theoretic notions, such as Shannon entropy~\cite{BKR09,CCH11,PMTP12,S09} or min-entropy~\cite{BKR09,MS11,PMTP12,S09}. In this work, we focus on computing Shannon entropy. 

The information-theoretic underpinnings of QIF analyses allow an end-user to link the computed quantities with the probability of an adversary successfully guessing a secret, or the worst-case computational effort required for the adversary to infer the underlying confidential information. Consequently, QIF has been applied in diverse use-cases such as software side-channel detection~\cite{KB07}, inferring search-engine queries through auto-complete responses sizes~\cite{CWWZ10}, and measuring the tendency of Linux to leak TCP-session sequence numbers~\cite{ZQRZ18}. 

In our experiments, we focus on demonstrating that we can compute the entropy for programs modeled by Boolean formulas; nevertheless, our techniques are general and can be extended to other models such as automata-based frameworks. Let a formula $\varphi (U,V)$ capture the relationship between $U$ and $V$ such that for every valuation to $U$ there is at most one valuation to $V$ such that $\varphi$ is satisfied; one can view $U$ as the set of inputs and $V$ as the set of outputs. Let $m = |V|$ and $n = |U|$. 
Let $p$ be a probability distribution over $\{0,1\}^{V}$ such that for every assignment $\sigma$ to $V$, i.e., $\sigma: V \mapsto \{0,1\}$, we have $p_{\sigma} = \frac{|\satisfying{\varphi(V \mapsto \sigma)}|}{|\satisfying{\varphi)}\downarrow_{U}|}$, where $\satisfying{\varphi(V \mapsto \sigma)}$ denotes the set of solutions of $\varphi(V \mapsto \sigma)$ and $\satisfying{\varphi)}\downarrow_{U}$ denotes the set of solutions of $\varphi$ projected to $U$. Then, the \emph{entropy of $\varphi$} is $H(\varphi) =  \sum\limits_{\sigma  \in 2^{V}} p_{\sigma} \log \frac{1}{p_{\sigma}}$.  

Indeed, the problem of computing the entropy of a distribution sampled by a
given circuit is closely related to the {\sc EntropyDifference} problem
considered by Goldreich and Vadhan~\cite{GV99}, and shown to be SZK-complete.
We therefore do not expect to obtain polynomial-time algorithms for this problem.
The techniques that have been proposed to compute $H(\varphi)$ exactly compute $p_{\sigma}$ for each $\sigma$. Observe that computing $p_{\sigma}$ is equivalent to the problem of model counting, which seeks to compute the number of solutions of a given formula. Therefore, the exact techniques require $\mathcal{O}(2^m)$ model-counting queries~\cite{BPFP17,ESBB19,K12}; therefore, such techniques often do not scale for large values of $m$.  Accordingly, the state of the art often relies on sampling-based techniques that perform well in practice but can only provide lower or upper bounds on the entropy~\cite{KRB20,RKBB19}. As is often the case, techniques that only guarantee lower or upper bounds can output estimates that can be arbitrarily far from the ground truth. Thus, this setting is an appealing target for PAC-style, high-probability
multiplicative approximation guarantees. We remark that  K\"{o}pf and Rybalchenko~\cite{KR10} used Batu et al.'s~\cite{BDKR05} lower bounds to conclude that their scheme could not be improved without usage of structural properties of the program. In this context, our paper continues the direction alluded by  K\"{o}pf and Rybalchenko and designs the first efficient multiplicative approximation scheme by utilizing white-box access to the program.

Indeed, our algorithm obtains an estimate that is guaranteed to lie within a $(1\pm\varepsilon)$-factor of $H(\varphi)$ with confidence at least $1-\delta$. Once again, we stress that we obtain such a multiplicative estimate even when $H(\varphi)$ is very small, as in the case of a password-checker as described above. 

Sampling and counting satisfying assignments to formulas are, of course,
computationally intractable problems in the worst case. Nevertheless, systems
for solving these problems in practice have been developed, that frequently 
achieve reasonable performance in spite of their lack of running time 
guarantees~\cite{T06,SGRM18,AHT18,GSRM19,DV20}. Still, their invocation is relatively expensive; hence, the
situation is an excellent match to the property testing model, in which we
primarily count the number of such queries as the complexity measure of 
interest; we detail in Section~\ref{sec:experiments} how probability-revealing conditional sampling oracle, {\proc}, can be implemented with two calls to a model counter and one call to a sampler. 

We further observe that the knowledge of distribution $p$ defined by the underlying Boolean formula $\varphi$ allows us to reduce the number of queries to {\proc} from  $\mathcal{O}(\frac{m}{\varepsilon^2})$ to $\mathcal{O}(\frac{min(m,n)}{\varepsilon^2})$. Therefore, in contrast to the algorithms used in practice and prior work in the property testing literature, our algorithm makes only $\mathcal{O}(\frac{min(m,n)}{\varepsilon^2})$ counting and sampling queries even though the support of the distribution specified by $\varphi$ can be of size $2^m$.

 To illustrate the practical efficiency of our algorithm, we implement a prototype, {\EntropyEstimation}, that employs a state-of-the-art counter for model-counting queries, GANAK~\cite{SRSM19}, and SPUR~\cite{AHT18} for sampling queries. Our empirical analysis demonstrates that {\EntropyEstimation} is able to handle benchmarks that clearly lie beyond the reach of the exact techniques. We stress again that while we present {\EntropyEstimation} for programs modeled as a Boolean formula, our analysis applies other approaches, such as automata-based approaches, modulo access to the appropriate sampling and counting oracles.

\subsection{Organization of the rest of the paper}

The rest of the paper is organized as follows: we present the notations and  preliminaries in Section~\ref{sec:prelims}. 
Next, we present an overview of {\EntropyEstimation} including a detailed description of the algorithm and an analysis of its correctness in Section~\ref{sec:detail-overview}. We then describe our experimental methodology and discuss our results with respect to the accuracy and scalability of {\EntropyEstimation} in Section~\ref{sec:experiments}. Finally, we conclude in Section~\ref{sec:conclusion}.

\section{Preliminaries}\label{sec:prelims}
Let $\Omega$ be the universe and a probability distribution $D$ over $\Omega$ is a non-negative function $D: \Omega \mapsto [0,1]$ such that $\sum_{x \in \Omega} D(x) = 1$.
Let $D$ be a fixed distribution over $\Omega$ of size $2^m$.

Two oracles often studied in the property testing literature are conditioning, denoted by COND, and evaluation, denoted by EVAL. A conditioning oracle for a distribution $D$, COND, takes as input a set $S \subseteq \Omega$ and returns $x$ such that the probability $x$ is returned is $\frac{D(x)}{\sum_{y\in S} D(y)}$. An evaluation oracle for $D$, EVAL, respontds to a query $x \in \Omega$ with $D(x)$. 

A {\em probability-revealing conditional sampling} oracle for $D$, $\proc$, when queried with a set $S \subseteq \Omega$, returns a tuple $(x, D(x))$ such that $x \in S$ and the probability $x$ is returned is $\frac{D(x)}{\sum_{y \in S} D(y)}$. 
Note that access to the probability-revealing conditional sampling oracle, {\proc}, is indeed weaker than access to both COND and EVAL, as calling EVAL on the
$x$ returned by COND permits simulation of {\proc}, but {\proc} does not permit access to $D(x)$ for an arbitrary $x$.

\section{{\EntropyEstimation}: Efficient Estimation of $H(D)$}\label{sec:detail-overview}
In this section, we focus on the primary technical contribution of our work: an algorithm, called {\EntropyEstimation}, that returns  an  $(\varepsilon,\delta)$ estimate of $H(D)$. We first provide a detailed technical overview of the design of {\EntropyEstimation}  in Section~\ref{techover-sec}, then provide a detailed description of the algorithm, and finally, provide the accompanying technical analysis of the correctness and complexity of {\EntropyEstimation}. 

\subsection{Technical Overview}\label{techover-sec}
At a high level, {\EntropyEstimation} uses a median of means estimator, i.e., we first estimate $H(D)$ to within a $(1\pm\varepsilon)$-factor with probability at least $\frac{5}{6}$ by computing the mean of the underlying estimator and then take the median of many such estimates to boost the probability of correctness to $1-\delta$.  Recall $|\Omega|=2^m$. 

Let us consider a  random variable $Z$ over the domain $\Omega$ with distribution $D$ and consider the self-information function $g: \Omega \to [0,\infty)$, given by $g(x) = \log (\frac{1}{D(x)})$. 
Observe that the entropy $H(D) = \expect[g(Z)]$. Therefore, a simple estimator would be to sample $Z$ using our oracle and then estimate the expectation of $g(Z)$ by a sample mean. 
 In their seminal work, Batu et al.~\cite{BDKR05} observed that the variance of $g(Z)$, denoted by $\variance[g(Z)]$, can be at most $m^2$. The required number of sample queries, based on a straightforward analysis, would  be $\Theta\left(\frac{ \variance[g(Z)]}{\varepsilon^2 \cdot (\expect[g(Z)])^2}\right) = \Theta\left(\frac{\sum D(x) \log^2 \frac{1}{D(x)}}{(\sum D(x) \log \frac{1}{D(x)})^2}\right)$. However, $\expect[g(Z)] = H(D)$ can be arbitrarily close to $0$, and therefore, this does not provide a reasonable upper bound on the required number of samples. 

To address the lack of lower bound on $H(D)$, we observe that for $D$ to have $H(D) < 1$, there must exist $x_{high}$ such that $D(x_{high}) > \frac{1}{2}$. We then observe that given access to {\proc}, we can identify such a $x$ with high probability, thereby allowing us to consider the two cases separately: (A) $H(D) > 1$ and (B) $H(D) < 1$. Now, for case (A), we could use Batu et al's bound for $\variance[g(Z)]$ and obtain an estimator that would require $\Theta\left(\frac{ \variance[g(Z)]}{\varepsilon^2 \cdot (\expect[g(Z)])^2}\right)$ queries to {\proc}. It is worth remarking that the bound $\variance[g(Z)] \leq m^2$ is indeed tight as a uniform distribution over $D$ would achieve the bound. Therefore, we instead focus on the expression   $\frac{ \variance[g(Z)]}{(\expect[g(Z)])^2}$ and prove that for the case when $\expect[g(Z)]  = H(D) > h $, we can upper bound $\frac{ \variance[g(Z)]}{(\expect[g(Z)])^2}$ by $\frac{(1+o(1)) \cdot m}{h\cdot \varepsilon^2}$, thereby reducing the complexity from $m^2$ to $m$ (Observe that we have $H(D) > 1$, that is, we can take $h=1$). 

Now we return to the case (B)  wherein we have identified $x_{high}$ with $D(x_{high}) > \frac{1}{2}$. Let $r = D(x_{high})$ and $H_{rem} = \sum\limits_{y \in \Omega \setminus x_{high}} D(y) \log \frac{1}{D(y)}$. Note that $H(D) = r \log \frac{1}{r} + H_{rem}$. Therefore, we focus on estimating $H_{rem}$. To this end, we define a random variable $T$ that takes values in $\Omega \setminus \{x_{high}\}$ such that $\Pr[T = y] = \frac{D(y)}{1-r}$. Using the function $g$ defined above, we have $H_{rem} = (1-r) \cdot \expect[g(T)]$. Again, we have two cases, depending on whether $H_{rem} \geq 1$ or not; if it is, then we can bound the ratio $\frac{\variance[g(T)]}{\expect[g(T)]^2}$ similarly to case (A). If not, we observe that the denominator is
at least $1$ for $r\geq 1/2$. And, when $H_{rem}$ is so small, we can upper bound the numerator by $(1+o(1))m$, giving overall $\frac{ \variance[g(T)]}{(\expect[g(T)])^2}\leq (1+o(1)) \cdot\frac{1}{\varepsilon^2}\cdot m$. We can thus estimate $H_{rem} $ using the median of means estimator.

\subsection{Algorithm Description}

{\EntropyEstimation} takes a tolerance parameter $\varepsilon$, a confidence parameter $\delta$ as input, and returns an estimate $\hat{h}$ of the entropy $H(D)$, that is guaranteed to lie within a $(1\pm\varepsilon)$-factor of $H(D)$ with confidence at least $1-\delta$. Before presenting the technical details of {\EntropyEstimation}, we will first discuss the key subroutine {\SampleEst} in {\EntropyEstimation}. 

\begin{algorithm}[h]
	\caption{\label{algo:sampleest}$\SampleEst(\bar{S}, t, \delta)$ }
	\begin{algorithmic}[1]
		\State $\mathcal{L} \gets [\;]$
		\State $T \gets \lceil\frac{9}{2} \log\frac{2}{\delta}\rceil$ 	
		\For{$i = 1,\ldots,T$} \label{algo:sampleest:line:valueT}
		\State $est \gets 0$
		\For{$j = 1,\ldots,t$}\label{algo:sampleest:line:innerloop}
		\State $ (y,r) \gets \proc (\Omega\setminus\bar{S} )$ \label{algo:sampleest:line:proc}
		\State $est \gets est +  \log (1/r)$ \label{algo:sampleest:line:updateest}
		\EndFor 
		\State $\mathcal{L}.\mathrm{Append}(\frac{est}{t})$ \label{algo:sampleest:line:appendC}
		\EndFor 
		\State \Return $\mathrm{Median}(\mathcal{L})$~\label{algo:sampleest:line:returnC}
	\end{algorithmic}
\end{algorithm}

Algorithm~\ref{algo:sampleest} presents the subroutine {\SampleEst}, which  takes as input an element $x$; the number of required samples, $t$; and a confidence parameter $\delta$, and returns a median-of-means estimate of $H_{rem}$. Algorithm~\ref{algo:sampleest} starts off by computing the value of $T$, the required number of repetitions to ensure at least $1-\delta$ confidence for the estimate. The algorithm has two loops--- one outer loop (Lines~\ref{algo:sampleest:line:valueT}-\ref{algo:sampleest:line:appendC}), and one inner loop (Lines~\ref{algo:sampleest:line:innerloop}-\ref{algo:sampleest:line:updateest}).  The outer loop runs for $\lceil\frac{9}{2} \log(\frac{2}{\delta})\rceil$ rounds, where in each round, Algorithm~\ref{algo:sampleest} updates a list $\mathcal{L}$ with the mean estimate, \emph{est}. In the inner loop, in each round, Algorithm~\ref{algo:sampleest} updates the value of \emph{est}: Line~\ref{algo:sampleest:line:proc} invokes ${\proc}$ to draw sample from $D$ conditioned on the set $\Omega \setminus \{x\}$. 
At line~\ref{algo:sampleest:line:updateest}, \emph{est} is updated with $\log(\frac{1}{r})$, and at line~\ref{algo:sampleest:line:appendC}, the final \emph{est} is added to $\mathcal{L}$. Finally, at line~\ref{algo:sampleest:line:returnC}, Algorithm~\ref{algo:sampleest} returns the median of $\mathcal{L}$.

We now return to {\EntropyEstimation}; 
Algorithm~\ref{algo:entropyest} presents the proposed algorithmic framework {\EntropyEstimation}. 

\begin{algorithm}[htb]
	\caption{\label{algo:entropyest}$\EntropyEstimation(\varepsilon,\delta)$ }
	\begin{algorithmic}[1]	
		\State $m \gets \log |\Omega|$ %
		\For{$i = 1,\ldots,\lceil\log (10/\delta)\rceil$} \label{algo:entropyset:line:loopstart}
		
		\State $(x,r) \gets \proc (\Omega)$ 	\label{algo:entropyset:line:sample}
		\If{$r >  \frac{1}{2}$ }\label{algo:entropyset:line:checkr}
		\State $t \gets \frac{6}{\epsilon^2} \cdot (m+\log(m+\log m+2.5))$ \label{algo:entropyset:line:compute-t-inr}
		\State $\hat{h}_{rem} \gets \SampleEst(\{x\},t, 0.9\cdot\delta)$ \label{algo:entropyset:line:callsamplees-inr}
		\State $\hat{h} \gets (1-r) \hat{h}_{rem} + r \log (\frac{1}{r})$ \label{algo:entropyset:line:hath-inr}
		\State \Return $\hat{h}$ \label{algo:entropyset:line:return-inr}
		\EndIf
		\EndFor
		\State $t \gets \frac{6}{\epsilon^2} \cdot (n-1)$\label{algo:entropyset:line:compute-t}	
		\State $\hat{h} \gets \SampleEst(\varnothing,t,0.9\cdot\delta)$\label{algo:entropyset:line:callsamplees}
		\State \Return $\hat{h}$
	\end{algorithmic}
\end{algorithm}

 Algorithm~\ref{algo:entropyest} attempts to determine whether there exists $(x,D(x))$ such that $D(x) > 1/2$ or not by iterating over lines~\ref{algo:entropyset:line:loopstart}-\ref{algo:entropyset:line:return-inr} for $\lceil\log(10/\delta)\rceil$ rounds. Line~\ref{algo:entropyset:line:sample} draws a sample $(x,r=D(x))$. %
Line~\ref{algo:entropyset:line:checkr}  chooses one of the two paths based on the value of $r$:

\begin{enumerate}
	\item If the value of $r$  turns out to be greater than $1/2$, the value of required number of samples, $t$, is calculated as per the calculation shown at line~\ref{algo:entropyset:line:compute-t-inr}. At line~\ref{algo:entropyset:line:callsamplees-inr}, the subroutine {\SampleEst} is called to estimate $\hat{h}_{rem}$. Finally, it computes the estimate $\hat{h}$ at line~\ref{algo:entropyset:line:hath-inr}.
	
	\item If the value of $r$ is at most $1/2$ in every round, the number of samples we use, $t$, is calculated as per the calculation shown at line~\ref{algo:entropyset:line:compute-t}. At line~\ref{algo:entropyset:line:callsamplees}, the subroutine {\SampleEst} is called with appropriate arguments to compute the estimate $\hat{h}$.  
\end{enumerate}

\subsection{Theoretical Analysis}

\begin{theorem}\label{main-thm}
	Given access to {\proc} for a distribution $D$ with $m=\log |\Omega|\geq 2$, a tolerance parameter $\varepsilon > 0$, and confidence parameter $\delta > 0$, the algorithm {\EntropyEstimation} returns $\hat{h}$ such that 
	\begin{align*}
	\Pr\left[ (1-\varepsilon)H(D) \leq \hat{h} \leq (1+\varepsilon) H(D) \right] \geq 1-\delta
	\end{align*}
\end{theorem}

We first analyze the median-of-means estimator computed by {\SampleEst}.

\begin{lemma}\label{sampleest-lem}
	Given a set $\bar{S}$, access to {\proc} for a distribution $D$,
an accuracy parameter $\varepsilon > 0$,
a confidence parameter $\delta > 0$, 
and a batch size $t\in\mathbb{N}$ for which
\[
\frac{1}{t\epsilon^2}\cdot 
\left(
\frac{\sum_{y\in \Omega\setminus\bar{S}} D(y|\Omega\setminus\bar{S})
(\log\frac{1}{D(y)})^2}
{\left(\sum_{y\in \Omega\setminus\bar{S}} D(y|\Omega\setminus\bar{S})
\log\frac{1}{D(y)}\right)^2 }
-1
\right)
\leq 1/6 
\]
the algorithm {\SampleEst} returns an estimate $\hat{h}$ such that with
probability $1-\delta$, 
\begin{align*}
\hat{h}&\leq (1+\epsilon)\sum_{y\in\Omega\setminus\bar{S}}D(y|\Omega\setminus\bar{S})\log\frac{1}{D(y)} \text{ and}\\
\hat{h}&\geq (1-\epsilon)\sum_{y\in\Omega\setminus\bar{S}}D(y|\Omega\setminus\bar{S})\log\frac{1}{D(y)}.
\end{align*}
\end{lemma}
\begin{proof}
Let $R_{ij}$ be the random value taken by $r$ in the $i$th iteration of the
outer loop and $j$th iteration of the inner loop. We observe that 
$\{R_{ij}\}_{(i,j)}$ are a family of i.i.d.\ random variables.
Let $C_i=\sum_{j=1}^t\frac{1}{t}\log\frac{1}{R_{ij}}$ be the value appended to 
$C$ at the end of the $i$th iteration of the loop. 
Clearly $\expect{[C_i]}=\expect{[\log\frac{1}{R_{ij}}]}$.
Furthermore, we observe that by independence of the $R_{ij}$, 
\[
\variance{[C_i]} = \frac{1}{t}\variance{[\log\frac{1}{R_{ij}}]} = 
\frac{1}{t}(\expect{[(\log R_{ij})^2]}-\expect{[\log\frac{1}{R_{ij}}]}^2).
\]
By Chebyshev's inequality, now,
\begin{align*}
\Pr\left[|C_i-\expect{[\log\frac{1}{R_{ij}}]}|>
\epsilon\expect{[\log\frac{1}{R_{ij}}]}\right] &<
\frac{\variance{[C_i]}}{\epsilon^2\expect{[\log\frac{1}{R_{ij}}]}^2}\\
&=\frac{\expect{[(\log R_{ij})^2]}-\expect{[\log\frac{1}{R_{ij}}]}^2}
{t\cdot \epsilon^2\expect{[\log\frac{1}{R_{ij}}]}^2}\\
&\leq 1/6 
\end{align*}
by our assumption on $t$.

Let $L_i\in \{0,1\}$ be the indicator random variable for the event that 
$C_i<\expect{[\log\frac{1}{R_{ij}}]}-\epsilon\expect{[\log\frac{1}{R_{ij}}]}$,
and let $H_i\in \{0,1\}$ be the indicator random variable for the event that 
$C_i>\expect{[\log\frac{1}{R_{ij}}]}+\epsilon\expect{[\log\frac{1}{R_{ij}}]}$.
Similarly, since these are disjoint events, $B_i=L_i+H_i$ is also an indicator
random variable for the union. So long as $\sum_{i=1}^T L_i < T/2$ and
$\sum_{i=1}^T H_i < T/2$, we note that the value returned by {\SampleEst} is as
desired. By the above calculation, $\Pr[L_i=1]+\Pr[H_i=1]=\Pr[B_i=1]<1/6$, and 
we note that $\{(B_i,L_i,H_i)\}_i$ are a family of i.i.d.\ random variables. 
Observe that by Hoeffding's inequality,
\[
\Pr\left[\sum_{i=1}^T L_i \geq \frac{T}{6}+\frac{T}{3}\right]\leq \exp(-2T\frac{1}{9})=\frac{\delta}{2}
\]
and similarly $\Pr\left[\sum_{i=1}^T H_i \geq \frac{T}{2}\right]\leq\frac
{\delta}{2}$. Therefore, by a union bound, the returned value is adequate
with probability at least $1-\delta$ overall.
\end{proof}

The analysis of {\SampleEst} relied on a bound on the ratio of the first and
second moments of the self-information in our truncated distribution.
Suppose for all $x\in\Omega$, $D(x)\leq 1/2$. We observe that
then $H(D)\geq\sum_{y\in\Omega}D(y)\cdot 1=1$. In this case, we can
bound the ratio of the second moment to the square of the entropy as follows.

\begin{lemma}\label{y-bound-lem}
Let $D:\Omega\to [0,1]$ be given with 
$\sum_{y\in\Omega}D(y)\leq 1$ and 
\[
H=\sum_{y\in\Omega}D(y)\log \frac{1}{D(y)}\geq 1.
\]
Then
\[
\frac{\sum_{y\in\Omega}D(y)(\log D(y))^2}
{\left(\sum_{y\in\Omega}D(y)\log \frac{1}{D(y)}\right)^2}\leq \left(1+\frac{\log(m+\log m+1.1)}{m}\right)m.
\]
Similarly, if $H\leq 1$ and $m\geq 2$,
\[
\sum_{y\in\Omega}D(y)(\log D(y))^2\leq m+\log(m+\log m+2.5).
\]
\end{lemma}
Concretely, both cases give a bound that is at most $2m$ for $m\geq 3$;
$m=8$ gives a bound that is less than $1.5\times m$ in both cases, $m=64$ 
gives a bound that is less than $1.1\times m$, etc.
\begin{proof}
By induction on the size of the support $\supp$ of $D$, we'll show that 
when $H\geq 1$,
the ratio is at most $\log |\supp|+\log (\log|\supp|+\log\log |\supp|+1.1)$. 
Recall that we assume $|\Omega|=2^m$. The
base case is when there are only two elements ($m=1$), in which case both must
have $D(x)=1/2$, and the ratio is uniquely determined to be $1$. For the 
induction step, observe that whenever any subset of $\Omega$ takes value $0$
under $D$, this is equivalent to a distribution with smaller support, for 
which by induction hypothesis, we find the ratio is at most 
\begin{align*}
\log(|\supp|-1)+\log(\log(|\supp|-1)+\log\log(|\supp|-1)+1.1)\ &\\
<\log |\supp|+\log(\log |\supp|+\log\log |\supp|+1.1).&
\end{align*}
Consider any value of $H(D)=H$.
With the entropy fixed, we need only maximize the numerator of the ratio
Indeed, we've already ruled out a ratio of 
$|\supp|$ for solutions in which any of the $D(y)=0$ for $y\in\supp$, and
clearly we cannot have any $D(y) =1$, so we only need to consider interior 
points that are local optima. We use the method of Lagrange multipliers: for
some $\lambda$, all $D(y)$ must satisfy 
$\log^2D(y)+2\log D(y)-\lambda (\log D(y)-1)=0$, which has solutions
\[
\log D(y) = \frac{\lambda}{2}-1\pm\sqrt{(1-\frac{\lambda}{2})^2-\lambda}
=\frac{\lambda}{2}-1\pm\sqrt{1+\lambda^2/4}.
\]
We note that the second derivatives with respect to $D(y)$ are equal to
\[
\frac{2\log D(y)}{D(y)}+\frac{2-\lambda}{D(y)}
\]
which are negative iff $\log D(y) < \frac{\lambda}{2}-1$, hence we attain
local maxima only for the solution $\log D(y)=\frac{\lambda}{2}-1-\sqrt{1+\lambda^2/4}$. In other words, there is a single $D(y)$, which by the
entropy constraint, must satisfy $|\supp| D(y)\log\frac{1}{D(y)}=H$
which we'll show gives
\[
D(y) = \frac{H}{|\supp|(\log\frac{|\supp|}{H}+\log\log\frac{|\supp|}{H}+\rho)}
\]
for some $\rho\leq 1.1$. For $|\supp|=3$, we know $1\leq H\leq \log 3$, and we can
verify numerically that $\log\left(\frac{\log\frac{3}{H}+\log\log\frac{3}{H}+\rho}{\log\frac{3}{H}}\right)\in (0.42,0.72)$ for $\rho\in [0,1]$. Hence, by
Brouwer's fixed point theorem, such a choice of $\rho\in [0,1]$ exists. For 
$|\supp|\geq 4$, observe that $\frac{|\supp|}{H}\geq 2$, so $\log\left(\frac
{\log\frac{|\supp|}{H}+\log\log\frac{|\supp|}{H}}{\log\frac{|\supp|}{H}}\right)>0$. For $|\supp|=4$, $\log\left(\frac{\log\frac{4}{H}+\log\log\frac{4}{H}+\rho}{\log\frac{4}{H}}\right)\in [0,1]$, and similarly for all integer values of
$|\supp|$ up to 15, $\log\left(\frac{\log\frac{|\supp|}{H}+\log\log\frac
{|\supp|}{H}+1.1}{\log\frac{|\supp|}{H}}\right)<1.1$, so we can obtain
$\rho\in (0,1.1)$. Finally, for $|\supp|\geq 16$, we have 
$\frac{|\supp|}{H}\leq 2^{|\supp|/2H}$, and hence 
$\frac{\log\log\frac{|\supp|}{H}+\rho}{\log\frac{|\supp|}{H}}\leq 1$, so
\begin{align*}
|\supp|\frac{H(\log\frac{|\supp|}{H}+\log(\log\frac{|\supp|}{H}+\log\log\frac{|\supp|}{H}+\rho))}{|\supp|(\log\frac{|\supp|}{H}+\log\log\frac{|\supp|}{H}+\rho)}&\\
\leq H\frac{\log\frac{|\supp|}{H}+\log\log\frac{|\supp|}{H}+1}{\log\frac{|\supp|}{H}+\log\log\frac{|\supp|}{H}+\rho}&
\end{align*}
Hence it is clear that this gives $H$ for some $\rho\leq 1$.
Observe that for such a choice of $\rho$, using the substitution above, the 
ratio we attain is
\begin{align*}
\frac{|\supp|\cdot H}{H^2\cdot |\supp|(\log\frac{|\supp|}{H}+\log\log\frac{|\supp|}{H}+\rho)}\left(\log\frac{|\supp|(\log\frac{|\supp|}{H}+\log\log\frac{|\supp|}{H}+\rho)}{H}\right)^2&\\
=\frac{1}{H}(\log\frac{|\supp|}{H}+\log(\log\frac{|\supp|}{H}+\log\log\frac{|\supp|}{H}+\rho))
\end{align*}
which is monotone in $1/H$, so using the fact that $H\geq 1$, we find it is at
most
\[
\log |\supp|+\log(\log |\supp|+\log\log |\supp|+\rho)
\]
which, recalling $\rho < 1.1$, gives the claimed bound.

For the second part, observe that by the same considerations, when $H$ is fixed,
\[
\sum_{y\in\Omega}D(y)(\log D(y))^2=H\log\frac{1}{D(y)}
\]
for the unique choice of $D(y)$ for $m$ and $H$ as above, i.e., we will
show that for $m\geq 2$, it is
\[
H\left(\log\frac{|\Omega|}{H}+\log(\log\frac{|\Omega|}{H}+\log\log\frac{|\Omega|}{H}+\rho)\right)
\]
for some $\rho \in (0,2.5)$.
Indeed, we again consider the function
\[
f(\rho ) = \frac{\log(\log\frac{|\Omega|}{H}+\log\log\frac{|\Omega|}{H}+\rho)}{\log\log\frac{|\Omega|}{H}},
\]
and observe that for $|\Omega|/H > 2$, $f(0)>0$. Now, when $m\geq 2$ and
$H\leq 1$, $|\Omega|/H\geq 4$. We will see that the function 
$d(\rho )=f(\rho )-\rho$ has no critical points for $|\Omega|/H\geq 4$ and 
$\rho>0$, and hence its maximum is attained at the boundary, i.e., at 
$\frac{|\Omega|}{H}=4$, at which point we see that $f(2.5)<2.5$. So, for such
values of $\frac{|\Omega|}{H}$, $f(\rho)$ maps $[0,2.5]$ into $[0,2.5]$ and hence
by Brouwer's fixed point theorem again, for all $m\geq 4$ and $H\geq 1$
some $\rho\in (0,2.5)$ exists for which
$D(y) = \log\frac{|\Omega|}{H}+\log(\log\frac{|\Omega|}{H}+\log\log\frac{|\Omega|}{H}+\rho)$ gives $\sum_{y\in\Omega}D(y)\log\frac{1}{D(y)}=H$.

Indeed, $d'(\rho )=\frac{1}{\ln 2(\log\frac{|\Omega|}
{H}+\log\log\frac{|\Omega|}{H}+\rho)\log\log\frac{|\Omega|}{H}}-1$, which has a singularity
at $\rho = -\log\log\frac{|\Omega|}{H}-\log\log\frac{|\Omega|}{H}$, and
otherwise has a critical point at $\rho=\frac{\ln 2}{\log\log\frac{|\Omega|}{H}}-
\log\frac{|\Omega|}{H}-\log\log\frac{|\Omega|}{H}$. Since $\log\frac{|\Omega|}{H}
\geq 2$ and $\log\log\frac{|\Omega|}{H}\geq 1$ here, these are both clearly
negative. 

Now, we'll show that this expression (for $m\geq 2$) is maximized when $H=1$.
Observe first that the expression $H(m+\log\frac{1}{H})$ as a function of $H$
does not have critical points for $H\leq 1$: the derivative is $m+\log\frac{1}{H}-\frac{1}{\ln 2}$, so critical points require $H=2^{m-(1/\ln 2)}>1$.
Hence we see that this expression is maximized at the boundary, when $H=1$.
Similarly, the rest of the expression,
\[
H\log(m+\log\frac{1}{H}+\log(m+\log\frac{1}{H})+2.5)
\]
viewed as a function of $H$, only has critical points for
\[
\log(m+\log\frac{1}{H}+\log(m+\log\frac{1}{H})+2.5)=
\frac{\frac{1}{\ln 2}(1+\frac{1}{m+\log\frac{1}{H}})}{m+\log\frac{1}{H}+\log(m+\log\frac{1}{H})+2.5}
\]
i.e., it requires
\begin{align*}
(m+\log\frac{1}{H}+\log(m+\log\frac{1}{H})+2.5)\log(m+\log\frac{1}{H}+\log(m+\log\frac{1}{H})+2.5)&\\
=\frac{1}{\ln 2}(1+\frac{1}{m+\log\frac{1}{H}}).&
\end{align*}
But, the right-hand side is at most $\frac{3}{2\ln 2}<3$, while the left-hand 
side is at least $13$. Thus, it also has no critical points, and its maximum
is similarly taken at the boundary, $H=1$. Thus, overall, we find
\[
\sum_{y\in\Omega}D(y)(\log D(y))^2\leq m+\log(m+\log m+2.5)
\]
when $H\leq 1$ and $m\geq 2$.
\end{proof}

Although the assignment of probability mass used in the bound did not sum to 1,
nevertheless this bound is nearly tight. For any $\gamma>0$, and letting $H=1+
\Delta$ where $\Delta=\frac{1}{\log^{\gamma}(|\Omega|-2)}$, the following 
solution attains a ratio of $(1-o(1))m^{1-\gamma}$:
for any two $x^*_1,x^*_2\in \Omega$,
set $D(x^*_i)=\frac{1}{2}-\frac{\epsilon}{2}$ and set the
rest to $\frac{\epsilon}{|\Omega|-2}$, for $\epsilon$ chosen below. To obtain 
\begin{align*}
H&=2\cdot (\frac{1}{2}-\frac{\epsilon}{2})\log\frac{2}{1-\epsilon}+(|\Omega|-2)\cdot \frac{\epsilon}{|\Omega|-2}\log\frac{|\Omega|-2}{\epsilon}\\
&=(1-\epsilon)(1+\log(1+\frac{\epsilon}{1-\epsilon}))+\epsilon \log\frac{|\Omega|-2}{\epsilon}
\end{align*}
observe that since $\log(1+x)=\frac{x}{\ln 2}+\Theta(x^2)$, we will need to take
\begin{align*}
\epsilon &= \frac{\Delta}{\log(|\Omega|-2)+\log\frac{1-\epsilon}{\epsilon}-(1+\frac{1}{\ln 2})+\Theta(\epsilon^2)}\\
&= \frac{\Delta}{\log(|\Omega|-2)+\log\log(|\Omega|-2)+\log\frac{1}{\Delta}-(1+\frac{1}{\ln 2})-\frac{\epsilon}{\ln 2}+\Theta(\epsilon^2)}.
\end{align*}
For such a choice, we indeed obtain the ratio
\[
\frac{(1-\epsilon)\log^2\frac{2}{1-\epsilon}+\epsilon\log^2\frac{(|\Omega|-2)}{\epsilon}}{H^2}
\geq (1-o(1))m^{1-\gamma}.
\]

Using these bounds, we are finally ready to prove Theorem~\ref{main-thm}:

\begin{proof}
We first consider the case where no $x$ has 
$D(x) > 1/2$; here, the condition in line 6 of {\EntropyEstimation} never 
passes, so we return the value obtained by {\SampleEst} on line 12. Note that we
must have $H(D)\geq 1$ in this case. So, by
Lemma \ref{y-bound-lem},
\[
\frac{\sum_{x\in\Omega}D(x)(\log D(x))^2}
{\left(\sum_{x\in\Omega}D(x)\log \frac{1}{D(x)}\right)^2}
\leq \left(1+\frac{\log(m+\log m+1.1)}{m}\right)m
\]
and hence, by Lemma~\ref{sampleest-lem}, using 
$t\geq\frac{6\cdot (m+\log(m+\log m+1.1))-1)}{\varepsilon^2}$ 
suffices to ensure
that the returned $\hat{h}$ is satisfactory with probability $1-\delta$.

Next, we consider the case where some $x^*\in\Omega$ has 
$D(x^*) > 1/2$. Since the total probability is 1, there can be at most one
such $x^*$. So, in the distribution conditioned on $x\neq x^*$,
i.e., $\{D'(y)\}_{y\in\Omega}$ 
that sets $D'(x^*)=0$, and $D'(x)=\frac{D(x)}{1-D(x^*)}$ otherwise, 
we now need to show that $t$ satisfies
\[
\frac{1}{t\varepsilon^2}\left(\frac{\sum_{y\neq x^*}D'(y)(\log\frac{1}{(1-D(x^*))D'(y)})^2}{(\sum_{y\neq x^*}D'(y)\log\frac{1}{(1-D(x^*))D'(y)})^2}-1\right)<\frac{1}{6}
\]
to apply Lemma~\ref{sampleest-lem}. We first rewrite this expression. Letting
$H=\sum_{y\neq x^*}D'(y)\log\frac{1}{D'(y)}$ be the entropy
of this conditional distribution,
\begin{align*}
\frac{\sum_{y\neq x^*}D'(y)(\log\frac{1}{(1-D(x^*))D'(y)})^2}{(\sum_{y\neq x^*}D'(y)\log\frac{1}{(1-D(x^*))D'(y)})^2}
&=\frac{\sum_{y\neq x^*}D'(y)(\log\frac{1}{D'(y)})^2+2H\log\frac{1}{1-D(x^*)}+(\log\frac{1}{1-D(x^*)})^2}{(H+\log\frac{1}{1-D(x^*)})^2}\\
&=\frac{\sum_{y\neq x^*}D'(y)(\log\frac{1}{D'(y)})^2-H^2}{(H+\log\frac{1}{1-D(x^*)})^2}+1.
\end{align*}
There are now two cases depending on whether $H$ is greater than 1 or 
less than 1. When it is greater than 1, the first part of Lemma~\ref
{y-bound-lem} again gives
\[
\frac{\sum_{y \in \Omega}D'(y)(\log D'(y))^2}
{H^2}
\leq m+\log(m+\log m+1.1).
\]
When $H<1$, on the other hand, recalling $D(x^*)>1/2$ (so $\log\frac{1}{1-D(x^*)}\geq 1$), the second part of Lemma~\ref{y-bound-lem} gives that our expression is less than
\[
\frac{m+\log(m+\log m+2.5))-H^2}{(H+\log\frac{1}{1-D(x^*)})^2}
<m+\log(m+\log m+2.5).
\]
Thus, by Lemma~\ref{sampleest-lem},
\[
t\geq \frac{6\cdot (m+\log(m+\log m+2.5))}{\varepsilon^2}
\]
suffices to obtain $\hat{h}$ such that $\hat{h}\leq (1+\varepsilon)\sum_{y\neq x^*}
\frac{D(y)}{1-D(x^*)}\log\frac{1}{D(y)}$  and 
$\hat{h}\geq (1-\varepsilon)\sum_{y\neq x^*}
\frac{D(y)}{1-D(x^*)}\log\frac{1}{D(y)}$; 
hence we obtain such a $\hat{h}$ with probability at least $1-0.9\cdot \delta$ 
in line~\ref{algo:entropyset:line:hath-inr}, if we pass the test on line~\ref{algo:entropyset:line:checkr} of Algorithm~\ref{algo:entropyest},
thus identifying $\sigma^*$. Note that this value is adequate, so we need only
guarantee that the test on line~\ref{algo:entropyset:line:checkr} passes on one of the iterations with 
probability at least $1-0.1\cdot\delta$.

To this end, note that each sample($x$) on line~\ref{algo:entropyset:line:sample} is equal to $x^*$ with
probability $D(x^*) > \frac{1}{2}$ by hypothesis. Since
each iteration of the loop is an independent draw, the probability that the
condition on line~\ref{algo:entropyset:line:checkr} is not met after $\log\frac{10}{\delta}$ draws is less than
$(1-\frac{1}{2})^{\log\frac{10}{\delta}}=\frac{\delta}{10}$, as needed.
\end{proof}

\section{Application to Quantitative Information Flow}\label{sec:experiments}

We demonstrate the practicality of {\EntropyEstimation} via an application to quantitative information flow (QIF) analysis, a subject of increasing interest in the software engineering community. We begin by recalling the setting and defining notation that is often employed in the QIF community. We then discuss how the algorithm {\EntropyEstimation} can be implemented in practice and demonstrate the empirical effectiveness of {\EntropyEstimation}. 

\subsection{QIF Formulation}
\subsubsection*{Notation}
We use lower case letters (with subscripts) to denote propositional variables and upper case letters to denote a subset of variables. A literal is a boolean variable or its negation. 
We write $ V \varphi(U,V)$ to denote a formula over blocks of variables $U=\{u_1,\cdots,u_n\}$ and $V=\{v_1,\cdots,v_m\}$. For notational clarity, we use $\varphi$ to refer to $\varphi(U,V)$ when clear from the context. We denote $Vars(\varphi)$ as the set of variables appearing in $\varphi$, i.e. $Vars(\varphi) = U \cup V$

A \textit{satisfying assignment} or solution of a formula $\varphi$ is a mapping $\tau: Vars(\varphi) \rightarrow \{0,1\}$, on which the formula evaluates to True. We denote the set of all the solutions of $\varphi$ as $\satisfying{\varphi}$.  The problem of \emph{model counting} is to compute
$|\satisfying{\varphi}|$ for a given formula $\varphi$.  
An {\em uniform sampler} outputs a solution $y \in \satisfying{\varphi}$ such that $\prob[y \text{ is output}] = \frac{1}{|\satisfying{\varphi}|}$.

 For $\mathcal{P} \subseteq Vars(\varphi)$, $\tau_{\downarrow \mathcal{P}}$ represents the truth values of variables in $\mathcal{P}$ in a satisfying assignment $\tau$ of $\varphi$. For $\mathcal{P} \subseteq Vars(\varphi)$, we define $\satisfying{\varphi}_{\downarrow \mathcal{P}}$ as the set of solutions of $\varphi$ projected on $\mathcal{P}$.   Projected model counting and uniform sampling are defined
 analogously using $\satisfying{\varphi}_{\downarrow \mathcal{P}}$ instead of
 $\satisfying{\varphi}$, for a given projection
 set $\mathcal{P} \subseteq Vars(\varphi)$.

We say that $\varphi$ is a circuit formula if for all assignments $\tau_1, \tau_2 \in \satisfying{\varphi}$, we have ${\tau_1}_{\downarrow U}  = {\tau_2}_{\downarrow U} \implies \tau_1 = \tau_2$. 
For a circuit formula $\varphi(U,V)$ and for $\sigma: V \mapsto \{0,1\}$, we define $p_{\sigma} = \frac{|\satisfying{\varphi( V \mapsto \sigma)}|}{|\satisfying{\varphi}_{\downarrow U}|}$. 
Given a circuit formula $\varphi(U,V)$, we define the entropy of $\varphi$, denoted by $\entropy{\varphi}$ as follows:
$\entropy{\varphi} = - \sum_{\sigma \in 2^{V}} p_{\sigma} \log (p_{\sigma})$. 

\subsubsection*{Oracles based on Projected Counting and Sampling}

We now discuss how we can implement the oracles, EVAL, COND, and {\proc} given access to counters and uniform samplers. For $\sigma \in 2^{V}$, in order to compute $p_{\sigma}$,  we make two queries to a model counter to compute the numerator and denominator respectively. To compute the numerator, we invoke the counter on the formula $\varphi(V \mapsto \sigma)$ and we compute the denominator by invoking it on the formula $\varphi$ with the projection set $\mathcal{P}$ set to $U$. 

In order to sample $\sigma \in 2^V$ with probability $p_{\sigma}$, 
 given access to a uniform sampler, we can simply first sample $\tau \in \satisfying{\varphi}$ uniformly at random, and then output $\sigma =  \tau_{\downarrow V}$, which ensures $\Pr[\sigma \text{ is output}] = p_{\sigma}$. To condition on a set $S$ in Boolean formulas, we first construct a formula $\psi$ such that $\satisfying{\psi} = S$ and then invoke the sampler/counters on the formula $\varphi \wedge \psi$. 
 
Therefore, given access to a projected counter and sampler, we can implement {\proc} by a query to a uniform sampler followed by two queries to a model counter. Observe that the denominator in the computation of $p_{\sigma}$ is identical for all $\sigma$, therefore, from the view of practical efficiency, we can save the denominator in memory and reuse it for all the subsequent calls.

\subsubsection*{QIF Modeling}
A program $\Pi$ maps a set of controllable inputs ($C$) and secret inputs ($I$) to outputs ($O$) observable to an attacker. The attacker is interested in inferring $I$ based on the output $O$.  It is standard in the security community to employ circuit formulas to model such programs. To this end, we will focus on the case where the given program $\Pi$ is modeled using a circuit formula $\varphi(U,V)$.  

A straightforward adaptation of {\EntropyEstimation} would give us an approximation scheme with $\mathcal{O}(\frac{m}{\varepsilon^2} \log \frac{1}{\delta})$ model counting and sampling queries. While the developments in the past decade has led to significant improvements in the runtime performance of counters and samplers, it is of course still desirable to reduce the query complexity. 

We observe that in this model, in which we have access to $\varphi$, we can infer further properties of the distribution. In particular, for all  $\sigma \in 2^{V}$, we have $p_\sigma\geq 1/2^{|U|}$. This gives us another bound on the relative variance of the self information:

\begin{lemma}\label{x-bound-lem}
	Let $\{p_{\sigma}\in [1/2^{|U|},1]\}_{\sigma\in 2^V}$ 
	be given. 
	Then,
	\[
	\sum_{\sigma\in 2^V}p_{\sigma}(\log p_{\sigma})^2
	\leq |U|\sum_{\sigma\in 2^V}p_{\sigma}\log \frac{1}{p_{\sigma}}
	\]
\end{lemma}
\begin{proof}
	We observe simply that
	\[
	\sum_{\sigma\in 2^V}p_{\sigma}(\log p_{\sigma})^2
	\leq
	\log 2^{|U|}\sum_{\sigma\in 2^V}p_{\sigma}\log\frac{1}{p_{\sigma}}	
	=|U|\sum_{\sigma\in 2^V}p_{\sigma}\log \frac{1}{p_{\sigma}}.
	\]
\end{proof}

The above bound allow us to improve the sample complexity of {\EntropyEstimation} from $\mathcal{O}(\frac{m}{\varepsilon^2} \log \frac{1}{\delta})$ to $\mathcal{O}(\frac{\min(m,n)}{\varepsilon^2} \log \frac{1}{\delta})$. To this end, we make two modifications to {\EntropyEstimation}, as follows:
\begin{enumerate}
	\item line~\ref{algo:entropyset:line:compute-t-inr}  is modified to 		$t \gets \frac{6}{\epsilon^2} \cdot \min\left\{\frac{n}{2\log\frac{1}{1-r}},m+\log(m+\log m+2.5)\right\}$ 
	\item  line~\ref{algo:entropyset:line:compute-t} is modified to $t \gets \frac{6}{\epsilon^2} \cdot (\min\left\{{n},m+\log(m+\log m+1.1)\right\}-1)$
\end{enumerate}

\begin{corollary}\label{qif-cor}
	Given access to a circuit formula $\varphi$ with $|V|\geq 2$,
	a tolerance parameter $\varepsilon > 0$, and confidence parameter
	$\delta > 0$, the modification of {\EntropyEstimation} for circuit
	formulas returns $\hat{h}$ such that
	\begin{align*}
		\Pr\left[ (1-\varepsilon)\entropy{\varphi}\leq \hat{h}
		\leq (1+\varepsilon)\entropy{\varphi}\right] \geq 1-\delta
	\end{align*}
\end{corollary}

\begin{proof}
	The proof is very similar to the proof of Theorem~\ref{main-thm}, but we now
	make use of the bound in Lemma~\ref{x-bound-lem}: specifically, in the case
	where there is no $\sigma$ occurring with probability greater than $1/2$, Lemma\ref
	{x-bound-lem} together with Lemma~\ref{y-bound-lem} gives
	\[
	\frac{\sum_{\sigma\in 2^V}p_{\sigma}(\log p_{\sigma})^2}
	{\left(\sum_{\sigma\in 2^V}p_{\sigma}\log \frac{1}{p_{\sigma}}\right)^2}
	\leq \min\left\{|U|,\left(1+\frac{\log(|V|+\log|V|+1.1)}{|V|}\right)|V|\right\}
	\]
	and hence, by Lemma~\ref{sampleest-lem}, using
	$t\geq\frac{6\cdot \min\{|U|,|V|+\log(|V|+\log|V|+1.1)\}-1)}{\varepsilon^2}$
	indeed suffices to ensure
	that the returned $\hat{h}$ is satisfactory with probability $1-\delta$.
	
	Meanwhile, in the case where such a dominating element $\sigma^*$ exists, 
	letting $H=\sum_{\sigma\neq\sigma^*}p'_{\sigma}\log\frac{1}{p'_{\sigma}}$ be 
	the entropy of the distribution conditioned on avoiding $\sigma^*$, we note 
	that we had obtained
	\[
	\frac{\sum_{\sigma\neq\sigma^*}p'_{\sigma}(\log\frac{1}{(1-p_{\sigma^*})p'_{\sigma}})^2}{(\sum_{\sigma\neq\sigma^*}p'_{\sigma}\log\frac{1}{(1-p_{\sigma^*})p'_{\sigma}})^2}=\frac{\sum_{\sigma\neq\sigma^*}p'_{\sigma}(\log\frac{1}{p'_{\sigma}})^2-H^2}{(H+\log\frac{1}{1-p_{\sigma^*}})^2}+1.
	\]
	Lemma~\ref{x-bound-lem} now gives rather directly that this quantity is at most
	\[
	\frac{H|U|-H^2}{(H+\log\frac{1}{1-p_{\sigma^*}})^2}+1<\frac{|U|}{2\log\frac{1}{1-p_{\sigma^*}}}+1.
	\]
	Thus, by Lemma~\ref{sampleest-lem},
	\[
	t\geq \frac{6\cdot \min\{\frac{|U|}{2\log\frac{1}{1-p_{\sigma^*}}},|V|+\log(|V|+\log |V|+2.5)\}}{\varepsilon^2}
	\]
	now indeed suffices to obtain $\hat{h}$ such that $\hat{h}\leq (1+\varepsilon)\sum_{\sigma\neq\sigma^*}
	\frac{p_{\sigma}}{1-p_{\sigma^*}}\log\frac{1}{p_{\sigma}}$  and
	$\hat{h}\geq (1-\varepsilon)\sum_{\sigma\neq\sigma^*}
	\frac{p_{\sigma}}{1-p_{\sigma^*}}\log\frac{1}{p_{\sigma}}$. The rest of the
	argument is now the same as before.
\end{proof}

\subsection{Empirical Setup}

To evaluate the runtime performance of {\EntropyEstimation}, we implemented a prototype in Python that employs SPUR~\cite{AHT18} as a uniform sampler and GANAK~\cite{SRSM19} as a projected model counter.  We experimented with 96 Boolean formulas arising from diverse applications ranging from QIF benchmarks~\cite{FRS17}, plan recognition~\cite{SGM20}, bit-blasted versions of SMTLIB benchmarks~\cite{SGM20,SRSM19}, and QBFEval competitions~\cite{qbfeval17,qbfeval18}.
 The value of $n = |U|$ varies from 5 to 752 while the value of $m=|V|$ varies from 9 to 1447. 

In all of our experiments, the confidence parameter $\delta$ was set to $0.09$, and the tolerance parameter $\varepsilon$ was set to $0.8$.  All of our experiments were conducted on a high-performance computer cluster with each node consisting of a E5-2690 v3 CPU with 24 cores, and 96GB of RAM with a memory limit set to 4GB per core. Experiments were run in single-threaded mode on a single core with a timeout of 3000s. 
	
\paragraph{\textbf{Baseline:}} As our baseline, we implemented the following approach to compute the entropy exactly, which is representative of the current state of the art approaches~\cite{BPFP17,ESBB19,K12}. For each valuation $\sigma \in sol(\varphi)_{\downarrow V}$, we compute $p_{\sigma} = \frac{|sol(\varphi(V \mapsto \sigma))|}{|sol(\varphi)_{\downarrow U}|}$, where $|sol(\varphi(V \rightarrow \sigma))|$ is the count of satisfying assignments of $\varphi(V\mapsto\sigma)$, and $|sol(\varphi)_{\downarrow U}|$  represents the projected model count of $\varphi$ over $U$. Then, finally the entropy is computed as $\sum\limits_{\sigma \in 2^{V}} p_{\sigma} \log(\frac{1}{p_{\sigma}})$. Observe that from property testing perspective, this amount to only using EVAL oracle.

  Our evaluation demonstrates that {\tool} can scale to the formulas beyond the reach of the enumeration-based baseline approach. Within a given timeout of 3000 seconds, {\tool} is able to estimate the entropy for all the benchmarks, whereas the baseline approach could terminate only for 14 benchmarks. Furthermore, {\tool} estimated the entropy within the allowed tolerance for {\em all} the benchmarks.
\begin{table}[h]
	\centering

\begin{tabularx}{\textwidth}{lXXXXXXX}
		\toprule
		Benchmarks & $|U|$ & $|V|$ & \multicolumn{2}{c}{{Baseline}}&  & \multicolumn{2}{c}{\tool} \\
		\cmidrule[\lightrulewidth](lr){4-5} \cmidrule[\lightrulewidth](lr){7-8}
		&   &  & Time(s) & \makecell{EVAL\\queries} &   & Time(s) & \makecell{{\proc}\\queries}  \\ \midrule
		
		pwd-backdoor &  336 &  64 &  - &  1.84$\times  10^{19}$  &   &  5.41 &  1.25$\times  10^{2}$  \\  
		
		case31 & 13 & 40 & 201.02 & 1.02$\times  10^{3}$ &   & 125.36 & 5.65$\times  10^{2}$ \\

		case23 &  14 &  63    &  420.85 &  2.05$\times   10^{3}$ &  &  141.17 &  6.10$\times   10^{2}$  \\

		s1488\_15\_7 &  14 &  927   &  1037.71 &  3.84$\times   10^{3}$ & &  150.29 &  6.10$\times   10^{2}$    \\ 
		
		case58 &  19 &  77   &  3835.38 &  1.77$\times   10^{4}$ &   &  198.34 &  8.45$\times   10^{2}$ \\  
		
		bug1-fix-4 &  53 &  17 & 373.52 &  1.76$\times   10^{3}$ &   &  212.37 &  9.60$\times   10^{2}$   \\  
		
		s832a\_15\_7 &  23 &  670    &  - &  2.65$\times   10^{6}$ & &  247 &  1.04$\times   10^{3}$  \\  
		
		dyn-fix-1 &  40 &  48    &  - &  3.30$\times   10^{4}$ & &  252.2 &  1.83$\times   10^{3}$  \\ 
		
		s1196a\_7\_4 &  32 &  676    &  -&  4.22$\times   10^{7}$ &   &  343.68 &  1.46$\times   10^{3}$ \\
		
		backdoor-2x16 &  168 &  32    &  - &  1.31$\times   10^{5}$ &  &  405.7 &  1.70$\times   10^{3}$ \\
		
		CVE-2007 &  752 &  32 &  - &  4.29$\times   10^{9}$ &    &  654.54 &  1.70$\times   10^{3}$ \\
		
		subtraction32 &  65 &  218  & - &  1.84$\times   10^{19}$ &  &  860.88 &  3.00$\times   10^{3}$    \\ 
		
		case\_1\_b11\_1 &  48 &  292   &  - &  2.75$\times   10^{11}$ &  &  1164.36 &  2.20$\times   10^{3}$    \\ 
		
		s420\_new\_15\_7-1 &  235 &  116   &  - &  3.52$\times   10^{7}$ &  &  1187.23 &  5.72$\times   10^{3}$   \\

		case145 &  64 &  155   &  - &  7.04$\times   10^{13}$  &   &  1243.11 &  2.96$\times   10^{3}$\\   
		
		floor64-1 &  405 &  161   &  - &  2.32$\times   10^{27}$  &  &  1764.2 &  7.85$\times   10^{3}$\\

		s641\_7\_4 &  54 &  453    &  - &  1.74$\times   10^{12}$ &  &  1849.84 &  2.48$\times   10^{3}$  \\ 
		
		decomp64 &  381 &  191   &  - &  6.81$\times   10^{38}$ &  &  2239.62 &  9.26$\times   10^{3}$  \\

		squaring2 &  72 &  813   &  - &  6.87$\times   10^{10}$ &   &  2348.6 &  3.33$\times   10^{3}$ \\ 
		
		stmt5\_731\_730 &  379 &  311    &  - &  3.49$\times   10^{10}$ &    &  2814.58 &  1.49$\times   10^{4}$\\

		\bottomrule
		\\
	\end{tabularx}
	
	\caption{Entropy Estimation by {\tool} vs Baseline.  ``-''  represents that entropy could not be estimated due to timeout. Note that $m=|V|$ and $n=|U|$.}\label{tab:scalable}
\end{table}

\subsection{Scalability of {\tool}}

Table~\ref{tab:scalable} presents the performance of {\tool} vis-a-vis the baseline approach for 20 benchmarks. (The complete analysis for all of the benchmarks can be found in the appendix.) Column 1 of Table~\ref{tab:scalable} gives the names of the benchmarks, while columns 2 and 3 list the numbers of $U$ and $V$ variables. Columns 4 and 5 respectively present the time taken, number of samples used by baseline approach, and columns 6 and 7 present the same for {\tool}. The required number of samples for the baseline approach is $|sol(\varphi)_{\downarrow V}|$. We use ``-'' to represent timeout.

Table~\ref{tab:scalable} clearly demonstrates that {\tool} outperforms the baseline approach. As shown in Table~\ref{tab:scalable}, there are some benchmarks for which the projected model count on $V$ is greater than $10^{30}$, i.e., the baseline approach would need $10^{30}$ valuations to compute the entropy exactly. By contrast, the proposed algorithm {\tool} needed at most $\sim 10^{4}$ samples to estimate the entropy within the given tolerance and confidence. The number of samples required to estimate the entropy is reduced significantly with our proposed approach, making it scalable.

\subsection{Quality of Estimates}

There were only $14$ benchmarks out of $96$ for which the enumeration-based baseline approach finished within a given timeout of $3000$ seconds. Therefore, we compared the entropy estimated by {\tool} with the baseline for those 14 benchmarks only. Figure~\ref{fig:ratio-accuray} shows how accurate were the estimates of the entropy by {\tool}. The y-axis represents the observed error, which was calculated as $max(\frac{\text{Estimated}}{\text{Exact}}-1,\frac{\text{Exact}}{\text{Estimated}}-1)$, and the x-axis represents the benchmarks ordered in ascending order of observed error; that is, a bar at $x$ represents the observed error for a benchmark---the lower, the better.

\begin{figure}[htb]
	\centering
	\includegraphics[scale=0.50]{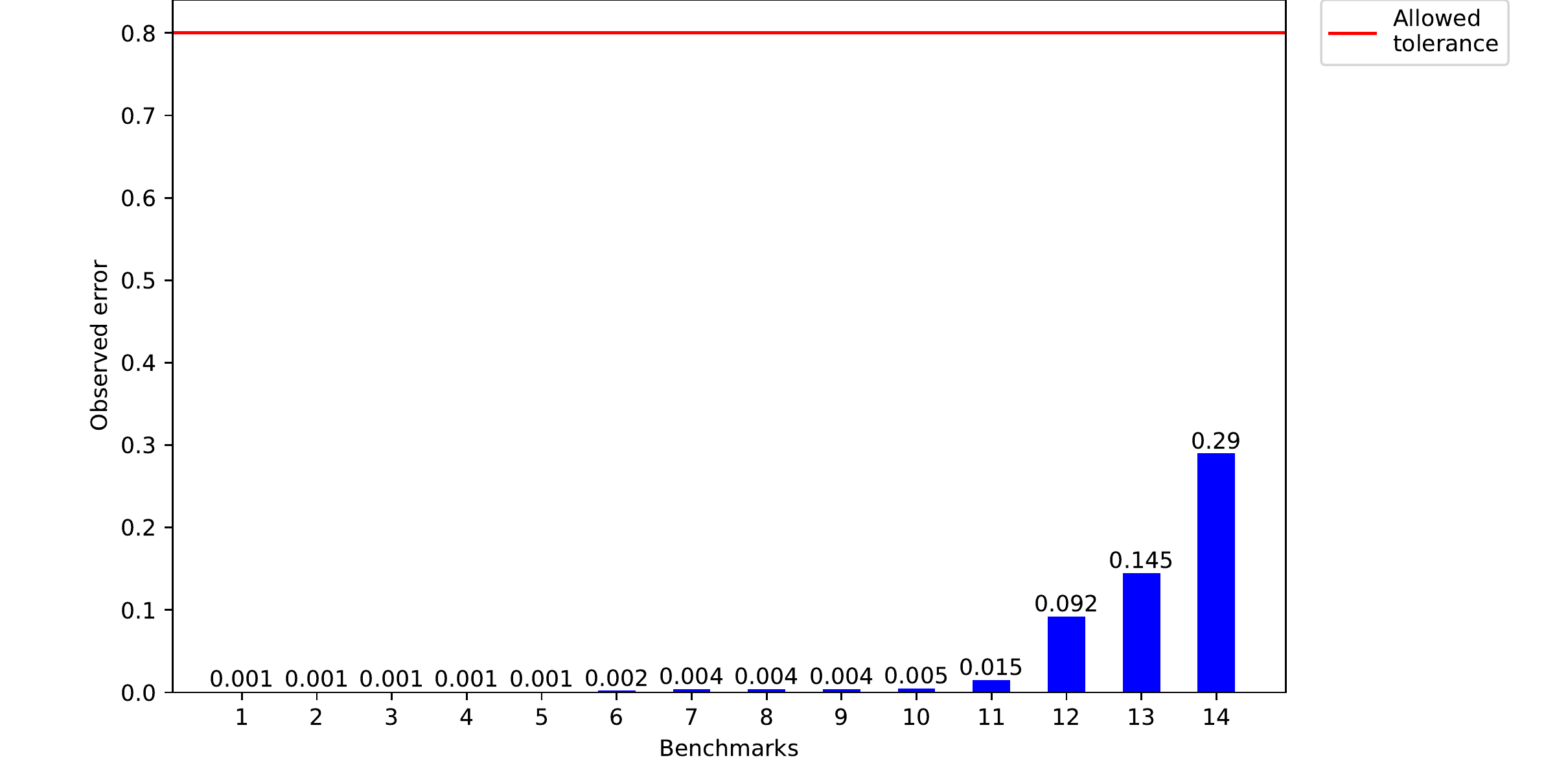}
	\caption{\label{fig:ratio-accuray} The accuracy of estimated entropy using {\tool} for $14$ benchmarks. $\varepsilon = 0.8, \delta = 0.1$.}
\end{figure}

The maximum allowed tolerance $(\varepsilon)$ for our experiments was set to $0.80$. The red horizontal line in Figure~\ref{fig:ratio-accuray} indicates this prescribed error tolerance.
We observe that for \emph{all} 14 benchmarks, {\tool} estimated the entropy within the allowed tolerance; in fact, the observed error was greater than $0.1$ for just 2 out of the 14 benchmarks, and the actual maximum error observed was $0.29$.
\paragraph{Alternative Baselines} As we discussed earlier, several other algorithms have been proposed for estimating the entropy. For example, Valiant and Valiant's algorithm~\cite{VV17} obtains an $\varepsilon$-additive approximation using $\mathcal{O}(\frac{2^m}{\varepsilon^2m})$ samples, and Chakraborty et al.~\cite{CFGM16} compute such approximations using $\mathcal{O}(\frac{m^7}{\varepsilon^8})$ samples. We stress that neither of these is exact, and thus could not be used to assess the accuracy of our method as presented in Figure~\ref{fig:ratio-accuray}. Moreover, based on Table~\ref{tab:scalable}, we observe that the number of sampling or counting calls that could be computed within the timeout was roughly $2\times 10^4$, where $m$ ranges between $10^1$--$10^3$. Thus, the method of Chakraborty et al., which would take $10^7$ or more samples on all benchmarks, would not be competitive with our method, which never used $2\times 10^4$ calls. The method of Valiant and Valiant, on the other hand, would likely allow a few more benchmarks to be estimated (perhaps up to a fifth of the benchmarks). Still, it would not be competitive with our technique except in the smallest benchmarks (for which the baseline required $<10^6$ samples, about a third of our benchmarks), since we were otherwise more than a factor of $m$ faster than the baseline.

\subsection{Beyond Boolean Formulas}\label{sec:beyond}
We now focus on the case where the relationship between $U$ and $V$ is modeled by an arbitrary relation $\mathcal{R}$ instead of a Boolean formula $\varphi$. As noted in Section~\ref{sec:qif-intro}, program behaviors are often modeled with other representations such as automata~\cite{ABB15,AEL+18,B19}. The automata-based modeling often has $U$ represented as the input to the given automaton $\mathcal{A}$ while every realization of $V$ corresponds to a state of $\mathcal{A}$. Instead of an explicit description of $\mathcal{A}$, one can rely on a symbolic description of $\mathcal{A}$. Two families of techniques are currently used to estimate the entropy. The first technique is to enumerate the possible {\em output} states and, for each such state $s$, estimate the number of strings accepted by $\mathcal{A}$ if $s$ was the only accepting state of $\mathcal{A}$. The other technique relies on uniformly sampling a string $\sigma$, noting the final state of $\mathcal{A}$ when run on $\sigma$, and then applying a histogram-based technique to estimate the entropy. 

In order to use the algorithm {\EntropyEstimation} one requires access to a sampler and model counter for automata; the past few years have witnessed the design of efficient counters for automata to handle string constraints. In addition, {\EntropyEstimation} requires access to a conditioning routine to implement the substitution step, i.e., $V \mapsto \sigma_{\downarrow V}$, which is easy to accomplish for automata via marking the corresponding state as a non-accepting state.    

\section{Conclusion}\label{sec:conclusion}
We thus find that the ability to draw conditional samples and obtain the
probability of those samples enables practical algorithms for estimating the
Shannon entropy, even in the low-entropy regime: we have only a linear
dependence on the number of bits to write down an element of the distribution,
and only a quadratic dependence on the approximation parameter $\epsilon$.
The constant factors are sufficiently small that the algorithm obtains good
performance on real benchmarks for computing the entropy of distributions
sampled by circuits, when the oracles are instantiated using existing methods
for model counting and sampling for formulas. Indeed, we find that this
setting is captured well by the property testing model, where the solvers for
these hard problems are treated as oracles and the number of calls is the
complexity measure of interest.

As mentioned in the introduction, the most interesting open question is 
whether or not $\mathcal{O}(\frac{m}{\epsilon^2}\log\frac{1}{\delta})$ is the
optimal number of such queries, even given an evaluation oracle with arbitrary
conditional samples. If the complexity could be reduced to $\mathcal{O}%
(\mathrm{poly}(\frac{\log m}{\epsilon})\log\frac{1}{\delta})$ as is the case 
for support size estimation (cf.~Acharya et al.~\cite{ACK18}) this would
further emphasize the power of conditional sampling. 

An important, related question is whether or not we can similarly efficiently 
obtain multiplicative estimates of the \emph{mutual information} between two 
variables in such a model, particularly in the low-information regime. Suppose
that we separate the outputs of the circuit into two parts, $Y$ and $Z$, where
$Z$ could represent some secret value, for example. (Since $Z$ can report part
of the input, this is more general than the problem of computing the mutual 
information with a secret portion of the input.) Observe that while we can 
separately estimate $H(Z)$ and $H(Z|Y)$ and compute an estimate of $I(Z;Y)$ 
from these, we obtain an error on the scale of $\epsilon\cdot \max\{H(Z),
H(Z|Y)\}$ which may be much larger than $\epsilon \cdot (H(Z)-H(Z|Y))$, which
is what would be required for a $(1+\epsilon)$-multiplicative approximation.

One further question suggested by this work is the relative power of an oracle
that reveals the conditional probability of the sample $\sigma$ obtained from
$D$ conditioned on $S$, instead of its probability under $D$. (Our algorithm
can certainly be adapted to such a model.) But as we noted in the 
introduction, whereas the oracle we considered in this work can be simulated
by a combination of the usual evaluation and conditional sampling oracles, it
seems unlikely that we could efficiently simulate this alternative 
probability-revealing conditional sampling oracle. This is because given a
single query to the probability-revealing conditional sampling oracle and one
additional query to an evaluation oracle (for the sampled value), we would 
be able to compute the \emph{exact} probability of the arbitrary event $S$, 
where this should require a large number of queries in the conditional 
sampling and evaluation model. We note that it seems that such an oracle could
equally well be implemented in practice in the circuit-formula setting we 
considered; does the additional power it grants allow us to do anything 
interesting?

Finally, an interesting direction for
future work on the practical side would be to extend {\EntropyEstimation} to 
handle other representations of programs such as automata-based models.

\small{\paragraph{Acknowledgments:} This work was supported in part by  National Research Foundation Singapore under its NRF Fellowship Programme[NRF-NRFFAI1-2019-0004 ], Ministry of Education Singapore Tier 2 grant [MOE-T2EP20121-0011],  NUS ODPRT grant [R-252-000-685-13], an Amazon Research Award, and  NSF awards IIS-1908287, IIS-1939677, and IIS-1942336. We are grateful to the anonymous reviewers for constructive comments to improve the paper. The computational work was performed on resources of the National Supercomputing Centre, Singapore: \url{https://www.nscc.sg}.}
\bibliographystyle{alpha}
\bibliography{ref}

\clearpage
\appendix

\section*{Detailed Experimental Analysis}
\begin{longtable}[h]{cccccccccc} 
     \caption{ Detailed results for all 96 benchmarks. Entropy Estimation by {\tool} vs Baseline.  $\delta: 0.09, \epsilon: 0.8$, and timeout 3000s. ``-''  represents that entropy could not be estimated due to timeout.}\\ %
         \midrule
         Benchmarks & $|U|$ & $|V|$ & \multicolumn{3}{c}{Baseline} &  & \multicolumn{3}{c}{{\tool}} \\[1em]
         \cmidrule{4-6} \cmidrule{8-10}
          &   &  & Time(s) & \makecell{ EVAL\\ Queries} & Entropy &  & Time(s) & \makecell{{\proc}\\ Queries} & Entropy \\ \hline \\
pwd-backdoor &  336 &  64 &  - &  1.84$\times  10^{19}$ &  - &   &  5.41 &  1.25$\times  10^{2}$ &  1.56$\times 10^{-19}$ \\  
case206 &  5 &  9 &  1.84 &  4.00$\times  10^{0}$ &  2.0 &   &  45.55 &  1.90$\times  10^{2}$ &  2.03 \\  
s27\_new\_7\_4 &  7 &  10 &  2.32 &  6.00$\times  10^{0}$ &  2.0 &   &  64.18 &  2.85$\times  10^{2}$ &  2.58 \\  
case31 &  13 &  40 &  201.02 &  1.02$\times  10^{3}$ &  10.0 &   &  125.36 &  5.65$\times  10^{2}$ &  10.04 \\  
case26 &  13 &  40 &  251.9 &  1.02$\times  10^{3}$ &  10.0 &   &  130.73 &  5.65$\times  10^{2}$ &  10.04 \\  
case27 &  13 &  39 &  195.94 &  1.02$\times  10^{3}$ &  10.0 &   &  133.08 &  5.65$\times  10^{2}$ &  10.04 \\  
case29 &  14 &  51 &  49.78 &  2.56$\times  10^{2}$ &  8.0 &   &  135.41 &  6.10$\times  10^{2}$ &  8.01 \\  
case23 &  14 &  63 &  420.85 &  2.05$\times  10^{3}$ &  11.0 &   &  141.17 &  6.10$\times  10^{2}$ &  11.01 \\  
s1488\_7\_4 &  14 &  858 &  2707.88 &  8.70$\times  10^{3}$ &  13.11 &   &  141.48 &  6.10$\times  10^{2}$ &  13.09 \\  
s1488\_15\_7 &  14 &  927 &  1037.71 &  3.84$\times  10^{3}$ &  11.92 &   &  150.29 &  6.10$\times  10^{2}$ &  11.91 \\  
bug1-fix-3 &  40 &  13 &  59.05 &  3.04$\times  10^{2}$ &  8.99 &   &  165.75 &  7.60$\times  10^{2}$ &  7.85 \\  
s298\_7\_4 &  17 &  206 &  - &  6.55$\times  10^{4}$ &  - &   &  166.93 &  7.50$\times  10^{2}$ &  16.0 \\  
case111 &  17 &  289 &  - &  1.64$\times  10^{4}$ &  - &   &  170.14 &  7.50$\times  10^{2}$ &  14.0 \\  
case113 &  18 &  291 &  - &  3.28$\times  10^{4}$ &  - &   &  176.34 &  8.00$\times  10^{2}$ &  15.06 \\  
case112 &  18 &  119 &  - &  3.28$\times  10^{4}$ &  - &   &  178.6 &  8.00$\times  10^{2}$ &  15.06 \\  
case4 &  18 &  85 &  - &  3.28$\times  10^{4}$ &  - &   &  188.1 &  8.00$\times  10^{2}$ &  15.06 \\  
bug1-fix-6 &  79 &  25 &  - &  6.23$\times  10^{4}$ &  - &   &  193.39 &  1.36$\times  10^{3}$ &  15.36 \\  
case64 &  19 &  74 &  - &  3.53$\times  10^{4}$ &  - &   &  194.06 &  8.45$\times  10^{2}$ &  15.13 \\  
case58 &  19 &  77 &  3835.38 &  1.77$\times  10^{4}$ &  14.11 &   &  198.34 &  8.45$\times  10^{2}$ &  14.13 \\  
case1 &  20 &  167 &  - &  6.55$\times  10^{4}$ &  - &   &  201.63 &  8.95$\times  10^{2}$ &  16.08 \\  
case53 &  21 &  111 &  - &  2.62$\times  10^{5}$ &  - &   &  207.03 &  9.40$\times  10^{2}$ &  18.05 \\  
bug1-fix-4 &  53 &  17 &  373.52 &  1.76$\times  10^{3}$ &  10.41 &   &  212.37 &  9.60$\times  10^{2}$ &  10.36 \\  
case46 &  22 &  154 &  - &  6.55$\times  10^{4}$ &  - &   &  214.78 &  9.85$\times  10^{2}$ &  16.01 \\  
case51 &  21 &  111 &  - &  2.62$\times  10^{5}$ &  - &   &  221.67 &  9.40$\times  10^{2}$ &  18.05 \\  
case54 &  23 &  180 &  - &  5.24$\times  10^{5}$ &  - &   &  231.67 &  1.04$\times  10^{3}$ &  19.07 \\  
s344\_7\_4 &  24 &  191 &  - &  9.54$\times  10^{5}$ &  - &   &  242.45 &  1.08$\times  10^{3}$ &  18.72 \\  
s444\_15\_7 &  24 &  353 &  - &  4.13$\times  10^{6}$ &  - &   &  244.88 &  1.08$\times  10^{3}$ &  22.02 \\  
s444\_7\_4 &  24 &  284 &  - &  1.28$\times  10^{7}$ &  - &   &  245.37 &  1.08$\times  10^{3}$ &  23.65 \\  
s832a\_15\_7 &  23 &  670 &  - &  2.65$\times  10^{6}$ &  - &   &  247 &  1.04$\times  10^{3}$ &  21.33 \\  
s832a\_3\_2 &  23 &  583 &  - &  2.87$\times  10^{5}$ &  - &   &  250.17 &  1.04$\times  10^{3}$ &  18.19 \\  
dyn-fix-1 &  40 &  48 &  - &  3.30$\times  10^{4}$ &  - &   &  252.2 &  1.83$\times  10^{3}$ &  15.02 \\  
s526\_3\_2 &  24 &  341 &  - &  4.19$\times  10^{6}$ &  - &   &  257.56 &  1.08$\times  10^{3}$ &  22.04 \\  
case136 &  42 &  169 &  - &  5.50$\times  10^{11}$ &  - &   &  262.21 &  1.92$\times  10^{3}$ &  39.06 \\  
bug1-fix-5 &  66 &  21 &  2520.7 &  1.04$\times  10^{4}$ &  14.0 &   &  264.68 &  1.16$\times  10^{3}$ &  12.82 \\  
case122 &  27 &  287 &  - &  1.68$\times  10^{7}$ &  - &   &  272.19 &  1.22$\times  10^{3}$ &  24.02 \\  
case114 &  28 &  400 &  - &  1.71$\times  10^{7}$ &  - &   &  285.78 &  1.27$\times  10^{3}$ &  25.09 \\  
case115 &  28 &  400 &  - &  1.69$\times  10^{7}$ &  - &   &  295.23 &  1.27$\times  10^{3}$ &  25.09 \\  
case116 &  28 &  410 &  - &  1.69$\times  10^{7}$ &  - &   &  303.52 &  1.27$\times  10^{3}$ &  25.09 \\  
case57 &  32 &  256 &  - &  1.68$\times  10^{7}$ &  - &   &  324.25 &  1.46$\times  10^{3}$ &  24.03 \\  
s1196a\_7\_4 &  32 &  676 &  - &  4.22$\times  10^{7}$ &  - &   &  343.68 &  1.46$\times  10^{3}$ &  24.97 \\  
s1238a\_15\_7 &  32 &  741 &  - &  4.04$\times  10^{7}$ &  - &   &  343.85 &  1.46$\times  10^{3}$ &  24.88 \\  
s420\_new\_15\_7 &  34 &  317 &  - &  3.41$\times  10^{7}$ &  - &   &  352.52 &  1.55$\times  10^{3}$ &  24.83 \\  
s420\_new\_7\_4 &  34 &  278 &  - &  3.52$\times  10^{7}$ &  - &   &  357.88 &  1.55$\times  10^{3}$ &  24.8 \\  
s420\_new1\_15\_7 &  34 &  332 &  - &  3.52$\times  10^{7}$ &  - &   &  359.18 &  1.55$\times  10^{3}$ &  24.85 \\  
s420\_3\_2 &  34 &  260 &  - &  3.52$\times  10^{7}$ &  - &   &  366.98 &  1.55$\times  10^{3}$ &  24.86 \\  
case\_0\_b12\_1 &  37 &  390 &  - &  1.07$\times  10^{9}$ &  - &   &  390.01 &  1.69$\times  10^{3}$ &  30.04 \\  
backdoor-2x16-8 &  168 &  32 &  - &  1.31$\times  10^{5}$ &  - &   &  405.7 &  1.70$\times  10^{3}$ &  8.0 \\  
case133 &  42 &  169 &  - &  5.50$\times  10^{11}$ &  - &   &  410.72 &  1.92$\times  10^{3}$ &  39.06 \\  
case\_3\_b14\_3 &  40 &  264 &  - &  1.37$\times  10^{11}$ &  - &   &  421.57 &  1.83$\times  10^{3}$ &  37.04 \\  
case132 &  41 &  195 &  - &  2.10$\times  10^{6}$ &  - &   &  423.72 &  1.88$\times  10^{3}$ &  21.0 \\  
case\_1\_b14\_3 &  40 &  264 &  - &  1.37$\times  10^{11}$ &  - &   &  441.42 &  1.83$\times  10^{3}$ &  37.04 \\  
bug1-fix-8 &  105 &  33 &  - &  2.24$\times  10^{6}$ &  - &   &  446.27 &  1.75$\times  10^{3}$ &  20.32 \\  
case\_3\_b14\_1 &  45 &  193 &  - &  1.72$\times  10^{10}$ &  - &   &  467.15 &  2.06$\times  10^{3}$ &  34.04 \\  
case\_1\_b14\_1 &  45 &  193 &  - &  1.72$\times  10^{10}$ &  - &   &  481.67 &  2.06$\times  10^{3}$ &  34.04 \\  
s953a\_7\_4 &  45 &  488 &  - &  4.24$\times  10^{5}$ &  - &   &  493.29 &  2.06$\times  10^{3}$ &  18.58 \\  
case201 &  45 &  155 &  - &  6.71$\times  10^{7}$ &  - &   &  500.23 &  2.06$\times  10^{3}$ &  26.03 \\  
case121-1 &  48 &  243 &  - &  7.52$\times  10^{10}$ &  - &   &  517.84 &  2.20$\times  10^{3}$ &  35.78 \\  
case121 &  48 &  243 &  - &  7.52$\times  10^{10}$ &  - &   &  556.54 &  2.20$\times  10^{3}$ &  35.78 \\  
10.sk\_1\_46 &  47 &  1447 &  - &  4.75$\times  10^{4}$ &  - &   &  560.45 &  2.16$\times  10^{3}$ &  13.56 \\  
bug1-fix-9 &  118 &  37 &  - &  1.34$\times  10^{7}$ &  - &   &  579.96 &  1.94$\times  10^{3}$ &  22.85 \\  
CVE-2007-2875 &  752 &  32 &  - &  4.29$\times  10^{9}$ &  - &   &  654.54 &  1.70$\times  10^{3}$ &  32.01 \\  
case53-1 &  75 &  57 &  62.9 &  1.03$\times  10^{3}$ &  10.0 &   &  661.23 &  2.91$\times  10^{3}$ &  10.01 \\  
case39 &  65 &  180 &  - &  3.60$\times  10^{16}$ &  - &   &  685.54 &  3.00$\times  10^{3}$ &  55.0 \\  
case106 &  60 &  144 &  - &  4.40$\times  10^{12}$ &  - &   &  710.96 &  2.77$\times  10^{3}$ &  42.07 \\  
case40 &  65 &  180 &  - &  3.60$\times  10^{16}$ &  - &   &  712.02 &  3.00$\times  10^{3}$ &  55.0 \\  
bug1-fix-10 &  131 &  41 &  - &  8.06$\times  10^{7}$ &  - &   &  729.35 &  2.14$\times  10^{3}$ &  25.36 \\  
case\_3\_b14\_1-1 &  165 &  73 &  - &  1.68$\times  10^{7}$ &  - &   &  832.09 &  3.68$\times  10^{3}$ &  24.02 \\  
subtraction32 &  65 &  218 &  - &  1.84$\times  10^{19}$ &  - &   &  860.88 &  3.00$\times  10^{3}$ &  64.0 \\  
case211 &  83 &  786 &  - &  1.21$\times  10^{24}$ &  - &   &  879.67 &  3.84$\times  10^{3}$ &  80.03 \\  
floor32 &  65 &  214 &  - &  6.86$\times  10^{14}$ &  - &   &  892.61 &  3.00$\times  10^{3}$ &  46.82 \\  
case146 &  64 &  155 &  - &  7.04$\times  10^{13}$ &  - &   &  920.28 &  2.96$\times  10^{3}$ &  46.03 \\  
case\_1\_b14\_1-1 &  145 &  93 &  - &  1.68$\times  10^{7}$ &  - &   &  1050.59 &  4.62$\times  10^{3}$ &  24.0 \\  
case\_1\_b11\_1 &  48 &  292 &  - &  2.75$\times  10^{11}$ &  - &   &  1164.36 &  2.20$\times  10^{3}$ &  38.03 \\  
ceiling32 &  65 &  277 &  - &  1.24$\times  10^{15}$ &  - &   &  1182.41 &  3.00$\times  10^{3}$ &  47.37 \\  
s420\_new\_15\_7-1 &  235 &  116 &  - &  3.52$\times  10^{7}$ &  - &   &  1187.23 &  5.72$\times  10^{3}$ &  24.78 \\  
decomp64-1 &  485 &  87 &  - &  6.81$\times  10^{38}$ &  - &   &  1232.71 &  4.34$\times  10^{3}$ &  33.01 \\  
case145 &  64 &  155 &  - &  7.04$\times  10^{13}$ &  - &   &  1243.11 &  2.96$\times  10^{3}$ &  46.03 \\  
dyn-fix-2 &  113 &  92 &  - &  8.45$\times  10^{6}$ &  - &   &  1337.35 &  4.58$\times  10^{3}$ &  23.02 \\  
floor32-1 &  150 &  129 &  - &  6.86$\times  10^{14}$ &  - &   &  1414.97 &  6.34$\times  10^{3}$ &  46.43 \\  
ceiling32-1 &  213 &  129 &  - &  1.10$\times  10^{15}$ &  - &   &  1642.07 &  6.34$\times  10^{3}$ &  47.09 \\  
subtraction64 &  129 &  442 &  - &  3.40$\times  10^{38}$ &  - &   &  1670 &  6.00$\times  10^{3}$ &  128.0 \\  
case116-1 &  264 &  174 &  - &  1.69$\times  10^{7}$ &  - &   &  1750.92 &  8.46$\times  10^{3}$ &  24.0 \\  
floor64-1 &  405 &  161 &  - &  2.32$\times  10^{27}$ &  - &   &  1764.2 &  7.85$\times  10^{3}$ &  86.71 \\  
case114-1 &  255 &  173 &  - &  1.71$\times  10^{7}$ &  - &   &  1799.49 &  8.42$\times  10^{3}$ &  24.02 \\  
stmt16\_818\_819 &  185 &  260 &  - &  1.03$\times  10^{10}$ &  - &   &  1823.11 &  8.62$\times  10^{3}$ &  24.96 \\  
squaring4 &  72 &  819 &  - &  6.87$\times  10^{10}$ &  - &   &  1825.71 &  3.33$\times  10^{3}$ &  36.02 \\  
s641\_7\_4 &  54 &  453 &  - &  1.74$\times  10^{12}$ &  - &   &  1849.84 &  2.48$\times  10^{3}$ &  37.89 \\  
case115-1 &  237 &  191 &  - &  1.69$\times  10^{7}$ &  - &   &  1922.39 &  9.26$\times  10^{3}$ &  23.99 \\  
subtraction64-1 &  409 &  162 &  - &  4.86$\times  10^{31}$ &  - &   &  2051.3 &  7.90$\times  10^{3}$ &  100.3 \\  
decomp64 &  381 &  191 &  - &  6.81$\times  10^{38}$ &  - &   &  2239.62 &  9.26$\times  10^{3}$ &  63.0 \\  
squaring2 &  72 &  813 &  - &  6.87$\times  10^{10}$ &  - &   &  2348.6 &  3.33$\times  10^{3}$ &  36.02 \\  
squaring1 &  72 &  819 &  - &  6.87$\times  10^{10}$ &  - &   &  2367.74 &  3.33$\times  10^{3}$ &  36.02 \\  
stmt124\_966\_965 &  393 &  310 &  - &  3.49$\times  10^{10}$ &  - &   &  2551.82 &  1.49$\times  10^{4}$ &  25.87 \\  
squaring6 &  72 &  813 &  - &  6.87$\times  10^{10}$ &  - &   &  2721.73 &  3.33$\times  10^{3}$ &  36.02 \\  
stmt9\_445\_446 &  352 &  306 &  - &  8.60$\times  10^{10}$ &  - &   &  2792.86 &  1.47$\times  10^{4}$ &  32.45 \\  
stmt5\_731\_730 &  379 &  311 &  - &  3.49$\times  10^{10}$ &  - &   &  2814.58 &  1.49$\times  10^{4}$ &  25.97 \\

\bottomrule

\end{longtable}

\end{document}